\documentclass[12pt]{article}  

\usepackage{amssymb}
\usepackage{amsmath}
\usepackage{amsthm}
\usepackage[usenames]{color}
\usepackage{colortbl}
\usepackage{graphicx}
\usepackage{hyperref}
\usepackage{fullpage}

\usepackage{lmodern}
\usepackage[T1]{fontenc}
\usepackage[utf8]{inputenc}

\newcommand{\eqdef}{\mathbin{\stackrel{\rm def}{=}}}

\newtheorem{theorem}{Theorem}
\newtheorem{lemma}[theorem]{Lemma}

\newtheorem{corollary}[theorem]{Corollary}

\newcommand{\uhard}{\mathcal{U}_{hard}}
\newcommand{\proj}{\mathbf{proj}}
\newcommand{\inprod}[1]{\left\langle #1 \right\rangle}
\newcommand{\eps}{\varepsilon}

\newcommand{\R}{{\mathbb{R}}}
\newcommand{\SE}{\Pi}
\DeclareMathOperator*{\E}{\mathbb{E}}
\let\Pr\relax
\DeclareMathOperator*{\Pr}{\mathbb{P}}
\DeclareMathOperator{\Var}{Var}

\DeclareMathOperator{\poly}{poly}

\DeclareMathOperator{\argmin}{argmin}
\DeclareMathOperator{\nnz}{nnz}

\newcommand{\CorollaryName}[1]{\label{cor:#1}}

\newcommand{\EquationName}[1]{\label{eq:#1}}

\newcommand{\LemmaName}[1]{\label{lem:#1}}

\newcommand{\TheoremName}[1]{\label{thm:#1}}

\newcommand{\Corollary}[1]{Corollary~\ref{cor:#1}}

\newcommand{\Equation}[1]{Eq.\:\eqref{eq:#1}}

\newcommand{\Lemma}[1]{Lemma~\ref{lem:#1}}

\newcommand{\Theorem}[1]{Theorem~\ref{thm:#1}}

\begin{document}

\title{Lower Bounds for Oblivious Subspace Embeddings}
\author{Jelani Nelson\thanks{Harvard University. \texttt{minilek@seas.harvard.edu}. This work was done while the author was a member at the Institute for Advanced Study, supported by NSF CCF-0832797 and NSF DMS-1128155.}\and Huy L. Nguy$\tilde{\hat{\mbox{e}}}$n\thanks{Princeton
  University. \texttt{hlnguyen@princeton.edu}. Supported in part by
  NSF CCF-0832797 and a Gordon Wu fellowship.}}

\maketitle

\begin{abstract}
An {\em oblivious subspace embedding (OSE)} for some $\eps,\delta\in (0,1/3)$ and $d\le m\le n$ is a distribution $\mathcal{D}$ over $\R^{m\times n}$ such that for any linear subspace $W\subset \R^n$ of dimension $d$,
$$\Pr_{\Pi\sim\mathcal{D}}\left(\forall x\in W,\ (1-\eps)\|x\|_2\le \|\Pi x\|_2 \le (1+\eps)\|x\|_2 \right) \ge 1 - \delta .$$
We prove that any OSE with $\delta < 1/3$ must have $m = \Omega((d + \log(1/\delta))/\eps^2)$, which is optimal. Furthermore, if every $\Pi$ in the support of $\mathcal{D}$ is sparse, having at most $s$ non-zero entries per column, then we show tradeoff lower bounds between $m$ and $s$.
\end{abstract}

\section{Introduction}
A {\em subspace embedding} for some $\eps\in(0,1/3)$ and linear subspace $W$ is a matrix $\SE$ satisfying
$$\forall x\in W,\ (1-\eps)\|x\|_2 \le \|\SE x\|_2 \le (1+\eps)\|x\|_2 . $$

An {\em oblivious subspace embedding (OSE)} for some $\eps,\delta\in (0,1/3)$ and integers $d\le m\le n$ is a distribution $\mathcal{D}$ over $\R^{m\times n}$ such that for any linear subspace $W\subset \R^n$ of dimension $d$,
\begin{equation}
\Pr_{\SE\sim\mathcal{D}}\left(\forall x\in W,\ (1-\eps)\|x\|_2\le \|\SE x\|_2 \le (1+\eps)\|x\|_2 \right) \ge 1 - \delta . \EquationName{ose}
\end{equation}
That is, for any linear subspace $W\subset \R^n$ of bounded dimension, a random $\SE$ drawn according to $\mathcal{D}$ is a subspace embedding for $W$ with good probability.

OSE's were first introduced in \cite{Sarlos06} and have since been used to provide fast approximate randomized algorithms for numerical linear algebra problems such as least squares regression \cite{CW13,MM13,NN13b,Sarlos06}, low rank approximation \cite{CW09,CW13,NN13b,Sarlos06}, minimum margin hyperplane and minimum enclosing ball \cite{PBMD13}, and approximating leverage scores \cite{MDMW12}. For example, consider the least squares regression problem: given $A\in\R^{n\times d}, b\in\R^n$, compute
$$ x^* = \argmin_{x\in\R^d} \|Ax - b\|_2 .$$
The optimal solution $x^*$ is such that $Ax^*$ is the projection of $b$ onto the column span of $A$. Thus by computing the singular value decomposition (SVD) $A = U\Sigma V^T$ where $U\in\R^{n\times r},V\in\R^{d\times r}$ have orthonormal columns and $\Sigma\in\R^{r\times r}$ is a diagonal matrix containing the non-zero singular values of $A$ (here $r$ is the rank of $A$), we can set $x^* = V\Sigma^{-1}U^T b$ so that $Ax^* = UU^T b$ as desired. Given that the SVD can be approximated in time $\tilde{O}(nd^{\omega - 1})$\footnote{We say $g = \tilde{O}(f)$ when $g = O(f\cdot \mathrm{polylog}(f))$.} \cite{DDH07} where $\omega < 2.373\ldots$ is the exponent of square matrix multiplication \cite{Williams12}, we can solve the least squares regression problem in this time bound. 

A simple argument then shows that if one instead computes
$$ \tilde{x} = \argmin_{x\in\R^d} \|\SE A x - \SE b\|_2 $$
for some subspace embedding $\SE$ for the $(d+1)$-dimensional subspace spanned $b$ and the columns of $A$, then $\|A\tilde{x} - b\|_2 \le (1+O(\eps))\|Ax^* - b\|_2$, i.e.\ $\tilde{x}$ serves as a near-optimal solution to the original regression problem. The running time then becomes $\tilde{O}(md^{\omega - 1})$, which can be a large savings for $m\ll n$, plus the time to compute $\SE A$ and $\SE b$ and the time to find $\SE$.

It is known that a random gaussian matrix with $m = O((d+\log(1/\delta))/\eps^2)$ is an OSE (see for example the net argument in Clarkson and Woodruff \cite{CW13} based on the Johnson-Lindenstrauss lemma and a net in \cite{AHK06}). While this leads to small $m$, and furthermore $\SE$ is oblivious to $A,b$ so that its computation is ``for free'', the time to compute $\SE A$ is $\tilde{O}(mnd^{\omega-2})$, which is worse than solving the original least squares regression problem. Sarl\'{o}s constructed an OSE $\mathcal{D}$, based on the fast Johnson-Lindenstrauss transform of Ailon and Chazelle \cite{AC09}, with the properties that (1) $m = \tilde{O}(d/\eps^2)$, and (2) for any vector $y\in\R^n$ and $\SE$ in the support of $\mathcal{D}$, $\SE y$ can be computed in time $O(n\log n)$ for any $\SE$ in the support of $\mathcal{D}$. This implies an approximate least squares regression algorithm running in time $O(nd\log n) + \tilde{O}(d^{\omega}/\eps^2)$. 

A recent line of work sought to improve the $O(nd\log n)$ term above to a quantity that depends only on the sparsity of the matrix $A$ as opposed to its ambient dimension. The works \cite{CW13,MM13,NN13b} give an OSE with $m = O(d^2/\eps^2)$ where every $\SE$ in the support of the OSE has only $s=1$ non-zero entry per column. The work \cite{NN13b} also showed how to achieve $m = O(d^{1+\gamma}/\eps^2), s = \poly(1/\gamma)/\eps$ for any constant $\gamma>0$. Using these OSE's together with other optimizations (for details see the reductions in \cite{CW13}), these works imply approximate regression algorithms running in time $O(\nnz(A) + (d^3\log d)/\eps^2)$ (the $s=1$ case), or $O_\gamma(\nnz(A)/\eps + d^{\omega+\gamma}/\eps^2)$ or $O_\gamma((\nnz(A) + d^2)\log(1/\eps) + d^{\omega+\gamma})$ (the case of larger $s$). Interestingly the algorithm which yields the last bound only requires an OSE with distortion $(1+\eps_0)$ for constant $\eps_0$, while still approximately the least squares optimum up to $1+\eps$.

As seen above we now have several upper bounds, though our understanding of lower bounds for the OSE problem is lacking. Any subspace embedding, and thus any OSE, must have $m\ge d$ since otherwise some non-zero vector in the subspace will be in the kernel of $\SE$ and thus not have its norm preserved. Furthermore, it quite readily follows from the works \cite{KMN11,MWY13} that any OSE must have $m = \Omega(\min\{n, \log(d/\delta)/\eps^2\})$ (see \Corollary{simple}). Thus the best known lower bound to date is $m = \Omega(\min\{n, d + \eps^{-2}\log(d/\delta)\})$, while the best upper bound is $m = O(\min\{n, (d + \log(1/\delta))/\eps^2\})$ (the OSE supported only on the $n\times n$ identity matrix is indeed an OSE with $\eps = \delta = 0$). We remark that although some problems can make use of OSE's with distortion $1+\eps_0$ for some constant $\eps_0$ to achieve $(1+\eps)$-approximation to the final problem, this is not always true (e.g.\ no such reduction is known for approximating leverage scores). Thus it is important to understand the required dependence on $\eps$.

\paragraph{Our contribution I:} We show that for any $\eps,\delta\in(0,1/3)$, any OSE with distortion $1+\eps$ and error probability $\delta$ must have $m = \Omega(\min\{n, (d + \log(1/\delta))/\eps^2\})$, which is optimal.

\bigskip

We also make progress in understanding the tradeoff between $m$ and $s$. The work \cite{NN13a} observed via a simple reduction to nonuniform balls and bins that any OSE with $s=1$ must have $m = \Omega(d^2)$. Also recall the upper bound of \cite{NN13b} of $m = O(d^{1+\gamma}/\eps^2), s = \poly(1/\gamma)/\eps$ for any constant $\gamma>0$.

\paragraph{Our contribution II:} We show that for $\delta$ a fixed constant and $n>100d^2$, any OSE with $m = o(\eps^2 d^2)$ must have $s = \Omega(1/\eps)$. Thus a phase transition exists between sparsity $s=1$ and super-constant sparsity somewhere around $m$ being $d^2$. We also show that for $m < d^{1+\gamma}$ and $\gamma \in ((10\log\log d)/(\alpha\log d), \alpha/4)$ and $2/(\eps \gamma) < d^{1-\alpha}$, for any constant $\alpha>0$, it must hold that $s = \Omega(\alpha/(\eps\gamma))$. Thus the $s = \poly(1/\gamma)/\eps$ dependence of \cite{NN13b} is correct (although our lower bound requires $m < d^{1+\gamma}$ as opposed to $m < d^{1+\gamma}/\eps^2$).

\bigskip

Our proof in the first contribution follows Yao's minimax principle combined with concentration arguments and Cauchy's interlacing theorem. Our proof in the second contribution uses a bound for nonuniform balls and bins and the simple fact that for {\em any} distribution over unit vectors, two i.i.d.\ samples are not negatively correlated in expectation. 

\subsection{Notation}
We let $O^{n\times d}$ denote the set of all $n\times d$ real matrices with orthonormal columns. For a linear subspace $W\subseteq\R^n$, we let $\proj_W:\R^n\rightarrow W$ denote the projection operator onto $W$. That is, if the columns of $U$ form an orthonormal basis for $W$, then $\proj_W x = UU^T x$. We also often abbreviate ``orthonormal'' as o.n. In the case that $A$ is a matrix, we let $\proj_A$ denote the projection operator onto the subspace spanned by the columns of $A$. Throughout this document, unless otherwise specified all norms $\|\cdot\|$ are $\ell_2\rightarrow\ell_2$ operator norms in the case of matrix argument, and $\ell_2$ norms for vector arguments. The norm $\|A\|_F$ denotes Frobenius norm, i.e.\ $(\sum_{i,j} A_{i,j}^2)^{1/2}$. For a matrix $A$, $\kappa(A)$ denotes the condition number of $A$, i.e.\ the ratio of the largest to smallest singular value. We use $[n]$ for integer $n$ to denote $\{1,\ldots,n\}$. We use $A\lesssim B$ to denote $A\le CB$ for some absolute constant $C$, and similarly for $A\gtrsim B$.
\section{Dimension lower bound}

Let $U\in O^{n\times d}$ be such that the columns of $U$ form an o.n.\ basis for a $d$-dimensional linear subspace $W$. Then the condition in \Equation{ose} is equivalent to all singular values of $\SE U$ lying in the interval $[1-\eps, 1+\eps]$. Let $\kappa(A)$ denote the condition number of matrix $A$, i.e.\ its largest singular value divided by its smallest singular value, so that for any such $U$ an OSE has $\kappa(\SE U) \le 1+\eps$ with probability $1-\delta$ over the randomness of $\SE$. Thus $\mathcal{D}$ being an OSE implies the condition
\begin{equation}
\forall U\in O^{n\times d} \Pr_{\SE\sim\mathcal{D}}\left(\kappa(\SE U) > 1+\eps\right) < \delta \EquationName{ose2}
\end{equation}


We now show a lower bound for $m$ in any distribution $\mathcal{D}$ satisfying \Equation{ose2} with $\delta < 1/3$. Our proof will use a couple lemmas. The first is quite similar to the Johnson-Lindenstrauss lemma itself. Without the appearance of the matrix $D$, it would follow from the the analyses in \cite{DG03,JL84} using Gaussian symmetry.

\begin{theorem}[Hanson-Wright inequality {\cite{HW71}}]
Let $g = (g_1,\ldots,g_n)$ be such that $g_i\sim\mathcal{N}(0,1)$ are independent, and let $B\in\R^{n\times n}$ be symmetric. Then for all $\lambda>0$,
$$ \Pr\left(\left|g^T B g - \mathrm{tr}(B)\right| > \lambda\right) \lesssim e^{-\min\left\{\lambda^2/\|B\|_F^2, \lambda/\|B\|\right\}} .$$
\end{theorem}

\begin{lemma}\LemmaName{jl-lemma}
Let $u$ be a unit vector drawn at random from $S^{n-1}$, and let $E\subset \R^n$ be an $m$-dimensional linear subspace for some $1\le m\le n$. Let $D\in\R^{n\times n}$ be a diagonal matrix with smallest singular value $\sigma_{min}$ and largest singular value $\sigma_{max}$. Then for any $0<\eps<1$
$$\Pr_u \left(\|\proj_E D u\|^2 \notin (\tilde{\sigma}^2 \pm \eps\sigma_{max}^2)\cdot\frac mn\right) \lesssim e^{-\Omega(\eps^2 m)} $$
for some $\sigma_{min} \le \tilde{\sigma}\le \sigma_{max}$.
\end{lemma}
\begin{proof}
Let the columns of $U\in O^{n\times m}$ span $E$, and let $u_i$ denote the $i$th row of $U$. Let the singular values of $D$ be $\sigma_1^2,\ldots,\sigma_n^2$. The random unit vector $u$ can be generated as $g/\|g\|$ for a multivariate Gaussian $g$ with identity covariance matrix. Then 
\begin{equation}
\|\proj_E D u\| = \frac 1{\|g\|} \cdot \|UU^T D g\| = \frac {\|U^T D g\|}{\|g\|}  . \EquationName{use-hw}
\end{equation}

We have
$$
\E \|U^T D g\|^2 = \E g^T DUU^T Dg = \mathrm{tr}(DUU^T D) = \sum_{i=1}^n \sigma_i^2\cdot \|u_i\|^2 = \tilde{\sigma}^2 \sum_i \|u_i\|^2 = \tilde{\sigma}^2 m,
$$
for some $\sigma_{min}^2 \le \tilde{\sigma}^2 \le \sigma_{max}^2$. Also 
$$\|DUU^T D\|_F^2 = \sum_{i=1}^n\sum_{j=1}^n \sigma_i^2\sigma_j^2 \inprod{u_i, u_j}^2 \le \sigma_{max}^4 \sum_{i,j} \inprod{u_i,u_j}^2 = \sigma_{max}^4 \sum_{i,j} m ,$$
and $\|DUU^T D\| \le \|D\|^2\cdot \|UU^T\| = \sigma_{max}^2$. Therefore by the Hanson-Wright inequality,
$$ \Pr\left(\left|\|U^T D g\|^2 - \tilde{\sigma}^2 m\right| > \eps \sigma_{max}^2 m\right) \lesssim e^{-\Omega(\min\{\eps^2 m, \eps m\})} = e^{-\Omega(\eps^2 m)} .$$
Similarly $\E\|g\|^2 = n$ and $\|g\|$ is also the product of a matrix with orthonormal columns (the identity matrix), a diagonal matrix with $\sigma_{min} = \sigma_{max} = 1$ (the identity matrix), and a multivariate gaussian. The analysis above thus implies
$$ \Pr\left(\left|\|g\|^2 - n\right| > \eps n\right) \lesssim e^{-\Omega(\eps^2 n)} .$$
Therefore with probability $1 - C(e^{-\Omega(\eps^2 n)} + e^{-\Omega(\eps^2 m)})$ for some constant $C>0$,
$$ \|\proj_E Du\|^2 = \frac {\|U^T D g\|^2}{\|g\|^2} = \frac{(\tilde{\sigma}^2 \pm \eps \sigma_{max}^2)m}{(1\pm \eps) n} = \frac{(\tilde{\sigma}^2 \pm O(\eps) \sigma_{max}^2)m}{n}$$
\end{proof}

We also need the following lemma, which is a special case of Cauchy's interlacing theorem.

\begin{lemma}\LemmaName{interlacing}
Suppose $A\in\R^{n\times m},A'\in\R^{(n+1)\times m}$ such that $n+1\le m$ and the first $n$ rows of $A,A'$ agree.Then the singular values of $A,A'$ interlace. That is, if the singular values of $A$ are $\sigma_1,\ldots,\sigma_n$ and those of $A'$ are $\beta_1,\ldots,\beta_{n+1}$, 
$$ \beta_1 \le \sigma_1 \le \beta_2 \le \sigma_2 \le \ldots \le \beta_n \le \sigma_n \le \beta_{n+1} .$$
\end{lemma}

Lastly, we need the following theorem and corollary, which follows from \cite{KMN11}. A similar conclusion can be obtained using \cite{MWY13}, but requiring the assumption that $d < n^{1-\gamma}$ for some constant $\gamma>0$. 

\begin{theorem}\TheoremName{use-yao}
Suppose $\mathcal{D}$ is a distribution over $\R^{m\times n}$ with the property that for any $t$ vectors $x_1,\ldots,x_t\in\R^n$,
$$ \Pr_{\SE\sim \mathcal{D}}\left(\forall i\in[t],\ (1-\eps)\|x_i\| \le \|\SE x_i\| \le (1+\eps)\|x_i\|\right) \ge 1 - \delta .$$
Then $m \gtrsim \min\left\{n, \eps^{-2}\log(t/\delta) \right\}$.
\end{theorem}
\begin{proof}
The proof uses Yao's minimax principle. That is, let $\mathcal{U}$ be an arbitrary distribution over $t$-tuples of vectors in $S^{n-1}$. Then 
\begin{equation}
\Pr_{(x_1,\ldots,x_t)\sim\mathcal{U}} \Pr_{\SE\sim\mathcal{D}}\left(\forall i\in[t],\ |\|\SE x_i\|^2 - 1| \le \eps \right) \ge 1-\delta .
\end{equation}
Switching the order of probabilistic quantifiers, an averaging argument implies the existence of a fixed matrix $\SE_0 \in\R^{m\times n}$ so that
\begin{equation}
\Pr_{(x_1,\ldots,x_t)\sim\mathcal{U}} \left(\forall i\in[t],\  |\|\SE_0 x\|^2 - 1| \le \eps\right) \ge 1 - \delta .\EquationName{impossible}
\end{equation}
The work \cite[Theorem 9]{KMN11} gave a particular distribution $\uhard$ for the case $t=1$ so that no $\SE_0$ can satisfy \Equation{impossible} unless $m \gtrsim \min\{n,\eps^{-2}\log(1/\delta)\}$. In particular, it showed that the left hand side of \Equation{impossible} is at most $1 - e^{-O(\eps^2 m + 1)}$ as long as $m \le n/2$ in the case $t=1$. For larger $t$, we simply let the hard distribution be $\uhard^{\otimes t}$, i.e.\ the $t$-fold product distribution of $\uhard$. Then the left hand side of \Equation{impossible} is at most $(1 - e^{-C(\eps^2 m + 1)})^t$. Let $\delta' = e^{-C(\eps^2 m + 1)}$. Thus $\mathcal{D}$ cannot satisfy the property in the hypothesis of the lemma if $(1-\delta')^t < 1 - \delta$. We have $(1 - \delta')^t \le e^{-t\delta'}$, and furthermore $e^{-x} = 1 - \Theta(x)$ for $0<x<1/2$. Thus we must have $t\delta' = O(\delta)$, i.e.\ $e^{-C(\eps^2 m + 1)} = \delta' = O(\delta/t)$. Rerranging terms proves the theorem.
\end{proof}

\begin{corollary}\CorollaryName{simple}
Any OSE distribution $\mathcal{D}$ over $\R^{m\times n}$ must have $m = \Omega(\min\{n, \eps^{-2}\log(d/\delta)\})$.
\end{corollary}
\begin{proof}
We have that for any $d$-dimensional subspace $W\subset\R^n$, a random $\SE\sim\mathcal{D}$ with probability $1-\delta$ simultaneously preserves norms of all $x\in W$ up to $1\pm\eps$. Thus for any set of $d$ vectors $x_1,\ldots,x_d\in\R^n$, a random such $\SE$ with probability $1-\delta$ simultaneously preserves the norms of these vectors since it even preserves their span. The lower bound then follows by \Theorem{use-yao}.
\end{proof}

Now we prove the main theorem of this section.

\begin{theorem}
Let $\mathcal{D}$ be any OSE with $\eps,\delta < 1/3$. Then $m = \Omega(\min\{n, d/\eps^2\})$.
\end{theorem}
\begin{proof}
We assume $d/\eps^2 \le cn$ for some constant $c>0$. Our proof uses Yao's minimax principle. Thus we must construct a distribution $\uhard$ such that
\begin{equation}
\Pr_{U\sim\uhard} \left(\kappa(\SE_0 U) > 1+\eps\right) < \delta .\EquationName{impossible2}
\end{equation}
cannot hold for any $\SE_0\in\R^{m\times n}$ which does not satisfy $m = \Omega(d/\eps^2)$. The particular $\uhard$ we choose is as follows: we let the $d$ columns of $U$ be independently drawn uniform random vectors from the sphere, post-processed using Gram-Schmidt to be orthonormal. That is, the columns of $U$ are an o.n.\ basis for a random $d$-dimensional linear subspace of $\R^n$.

Let $\SE_0 = LDW^T$ be the singular value decomposition (SVD) of $\SE_0$, i.e.\ $L\in O^{m\times n}, W\in O^{n\times n}$, and $D$ is $n\times n$ with $D_{i,i} \ge 0$ for all $1\le i\le m$, and all other entries of $D$ are $0$. Note that $W^T U$ is distributed identically as $U$, which is identically distributed as $W' U$ where $W'$ is an $n\times n$ block diagonal matrix with two blocks. The upper-left block of $W'$ is a random rotation $M\in O^{m\times m}$ according to Haar measure. The bottom-right block of $W'$ is the $(n-m)\times (n-m)$ identity matrix.  Thus it is equivalent to analyze the singular values of the matrix $LDW' U$. Also note that left multiplication by $L$ does not alter singular values, and the singular values of $DW' U$ and $D'MA^T U$ are identical, where $A$ is the $n\times m$ matrix whose columns are $e_1,\ldots,e_m$. Also $D'$ is an $m\times m$ diagonal matrix with $D'_{i,i} = D_{i,i}$. Thus we wish to show that if $m$ is sufficiently small, then  
\begin{equation}
\Pr_{M\sim O^{m\times m},U\sim\uhard}\left(\kappa(D'MA^TU) > 1+\eps\right) > \frac 13 \EquationName{main-target}
\end{equation}

Henceforth in this proof we assume for the sake of contradiction that $m \le c\cdot \min\{d/\eps^2, n\}$ for some small positive constant $c>0$. Also note that we may assume by \Corollary{simple} that $m = \Omega(\min\{n, \eps^{-2}\log(d/\delta)\})$.

Assume that with probability strictly larger than $2/3$ over the choice of $U$, we can find unit vectors $z_1,z_2$ so that $\|A^T U z_1\| / \|A^T U z_2\| > 1+\eps$. Now suppose we have such $z_1,z_2$. Define $y_1 = A^T U z_1/\|A^T U z_1\|, y_2 = A^T U z_2/\|A^T U z_2\|$. Then a random $M\in O^{m\times m}$ has the same distribution as $M'T$, where $M'$ is i.i.d.\ as $M$, and $T$ can be any distribution over $O^{m\times m}$, so we write $M = M' T$. $T$ may even depend on $U$, since $M' U$ will then still be independent of $U$ and a random rotation (according to Haar measure). Let $T$ be the $m\times m$ identity matrix with probability $1/2$, and $R_{y_1,y_2}$ with probability $1/2$ where $R_{y_1,y_2}$ is the reflection across the bisector of $y_1,y_2$ in the plane containing these two vectors, so that $R_{y_1,y_2}y_1 = y_2, R_{y_1,y_2} y_2 = y_1$. Now note that for any fixed choice of $M'$ it must be the case that $\|D'M'y_1\| \ge \|D'M'y_2\|$ or $\|D'M'y_2\|\ge \|D'M'y_1\|$. Thus $\|D'M'Ty_1\| \ge \|D'M'Ty_2\|$ occurs with probability $1/2$ over $T$, and the reverse inequality occurs with probability $1/2$. Thus for this fixed $U$ for which we found such $z_1,z_2$, over the randomness of $M',T$ we have $\kappa(D'M A^T U) \ge \|D'MA^T U z_1\| / \|D' MA^T U z_2\|$ is greater than $1+\eps$ with probability at least $1/2$. Since such $z_1,z_2$ exist with probability larger than $2/3$ over chioce of $U$, we have established \Equation{main-target}. It just remains to establish the existence of such $z_1,z_2$.

Let the columns of $U$ be $u^1,\ldots,u^d$, and define $\tilde{u}^i = A^T u^i$ and $\tilde{U} = A^T U$. Let $U_{-d}$ be the $n\times (d-1)$ matrix whose columns are $u^1,\ldots,u^{d-1}$, and let $\tilde{U}_{-d} = A^T U_{-d}$. Write $A = A^{\parallel} + A^{\perp}$, where the columns of $A^{\parallel}$ are the projections of the columns of $A$ onto the subspace spanned by the columns of $U_{-d}$, i.e.\ $A^{\parallel} = U_{-d}U_{-d}^T A$. Then
\begin{equation}
\|A^{\parallel}\|_F^2 = \|U_{-d}U_{-d}^T A\|_F^2 = \|\tilde{U}_{-d}\|_F^2 = \sum_{i=1}^{d-1} \sum_{r=1}^m (u^i_r)^2 . \EquationName{use-jl}
\end{equation}

By \Lemma{jl-lemma} with $D=I$ and $E = \mathrm{span}(e_1,\ldots,e_m)$, followed by a union bound over the $d-1$ columns of $U_{-d}$, the right hand side of \Equation{use-jl} is between $(1-C_1\eps)(d-1)m/n$ and $(1+C_1\eps)(d-1)m/n$ with probability at least $1 - C(d-1)\cdot e^{-C' C_1 \eps^2 m}$ over the choice of $U$. This is $1 - d^{-\Omega(1)}$ for $C_1>0$ sufficiently large since $m = \Omega(\eps^{-2}\log d)$.  Now, if $\kappa(\tilde{U}) > 1+\eps$ then $z_1,z_2$ with the desired properties exist. Suppose for the sake of contradiction that both $\kappa(\tilde{U}) \le 1+\eps$ and $(1-C_1\eps)(d-1)m/n \le \|\tilde{U}_{-d}\|_F^2 \le (1+C_1\eps)(d-1)m/n$. Since the squared Frobenius norm is the sum of squared singular values, and since $\kappa(\tilde{U}_{-d}) \le \kappa(\tilde{U})$ due to \Lemma{interlacing}, all the singular values of $\tilde{U}_{-d}$, and hence $A^{\parallel}$, are between $(1-C_2\eps)\sqrt{m/n}$ and $(1+C_2\eps)\sqrt{m/n}$. Then by the Pythagorean theorem the singular values of $A^{\perp}$ are in the interval $[\sqrt{1 - (1+C_2\eps)^2m/n}, \sqrt{1 - (1-C_2\eps)^2m/n}] \subseteq [1 - (1+C_3\eps)m/n, 1 - (1-C_3\eps)m/n]$.

Since the singular values of $\tilde{U}$ and $\tilde{U}^T$ are the same, it suffices to show $\kappa(\tilde{U}^T) > 1+\eps$. For this we exhibit two unit vectors $x_1,x_2$ with $\|\tilde{U}^T x_1\|/\|\tilde{U}^T x_2\| > 1+\eps$. Let $B\in O^{m\times d-1}$ have columns forming an o.n.\ basis for the column span of $AA^T U_{-d}$. Since $B$ has o.n.\ columns and $u^d$ is orthogonal to the column span of $U_{-d}$,
$$\|\proj_{\tilde{U}_{-d}} \tilde{u}^d\| = \|BB^T A^T u^d\| = \|B^T A^T u^d\| = \|B^T (A^{\perp})^T u^d\|.$$
Let $(A^{\perp})^T = C\Lambda E^T$ be the SVD, where $C\in\R^{m\times m},\Lambda\in\R^{m\times m}, E\in\R^{n\times m}$. As usual $C,E$ have o.n.\ columns, and $\Lambda$ is diagonal with all entries in $[1 - (1+C_3\eps)m/n, 1 - (1-C_3\eps)m/n]$. Condition on $U_{-d}$. The columns of $E$ form an o.n.\ basis for the column space of $A^{\perp}$, which is some $m$-dimensional subspace of the $(n-d+1)$-dimensional orthogonal complement of the column space of $U_{-d}$. Meanwhile $u^d$ is a uniformly random unit vector drawn from this orthogonal complement, and thus $\|E^T u_d\|^2 \in [(1-C_4\eps)^2m/(n-d+1), (1+C_4\eps)^2m/(n-d+1)] \subset [(1-C_5\eps)m/n, (1+C_5\eps)m/n]$ with probability $1-d^{-\Omega(1)}$ by \Lemma{jl-lemma} and the fact that $d \le \eps n$ and $m = \Omega(\eps^{-2}\log d)$. Note then also that $\|\Lambda E^T u^d\| = \|\tilde{u}^d\| = (1\pm C_6 \eps)\sqrt{m/n}$ with probability $1 - d^{-\Omega(1)}$ since $\Lambda$ has bounded singular values.

Also note $E^T u/\|E^T u\|$ is uniformly random in $S^{m-1}$, and also $B^T C$ has orthonormal rows since $B^T C C^T B = B^T B = I$, and thus again by \Lemma{jl-lemma} with $E$ being the row space of $B^T C$ and $D = \Lambda$, we have $\|B^T C \Lambda E^T u\| = \Theta(\|E^T u\|\cdot \sqrt{d/m}) = \Theta(\sqrt{d/n})$ with probability $1 - e^{-\Omega(d)}$.

We first note that by \Lemma{interlacing} and our assumption on the singular values of $\tilde{U}_{-d}$, $\tilde{U}^T$ has smallest singular value at most $(1+C_2\eps)\sqrt{m/n}$. We then set $x_2$ to be a unit vector such that $\|\tilde{U}^T x_2\| \le (1+C_2\eps)\sqrt{m/n}$.

It just remains to construct $x_1$ so that $\|\tilde{U}^T x_1\| > (1+\eps)(1+C_2\eps)\sqrt{m/n}$. To construct $x_1$ we split into two cases:

\paragraph{Case 1 ($m \le cd/\eps$):} In this case we choose 
$$ x_1 = \frac{\proj_{\tilde{U}_{-d}} \tilde{u}^d}{\|\proj_{\tilde{U}_{-d}} \tilde{u}^d\|} .$$

Then 
\begin{align*}
\|\tilde{U}^T x_1\|^2 &= \|\tilde{U}_{-d}^T x_1\|^2 + \inprod{\tilde{u}^d, x_1}^2\\
{}&\ge (1 - C_2\eps)^2\frac mn + \|\proj_{\tilde{U}_{-d}} \tilde{u}^d\|^2\\
{}&\ge (1 - C_2\eps)^2\frac mn + C\frac dn .\\
{}& \ge \frac mn\left((1 - C_2\eps)^2 + \frac Cc\eps\right)
\end{align*}

For $c$ small, the above is bigger than $(1+\eps)^2(1+C_2\eps)^2m/n$ as desired.

\paragraph{Case 2 ($cd/\eps \le m \le cd/\eps^2$):} In this case we choose
$$ x_1 = \frac 1{\sqrt{2}}\left[\frac{\overbrace{\proj_{\tilde{U}_{-d}} \tilde{u}^d}^{x^{\parallel}}}{\|\proj_{\tilde{U}_{-d}} \tilde{u}^d\|} + \frac{\overbrace{\proj_{\tilde{U}^{\perp}_{-d}} \tilde{u}^d}^{x^{\perp}}}{\|\proj_{\tilde{U}^{\perp}_{-d}} \tilde{u}^d\|}\right] .$$
Then
\allowdisplaybreaks
\begin{align}
\nonumber \|\tilde{U}^T x_1\|^2 &= \frac 12 \left\|\tilde{U}^T\left(\frac{x^{\parallel}}{\|x^{\parallel}\|} + \frac{x^{\perp}}{\|x^{\perp}\|}\right)\right\|^2\\
\nonumber {}& = \frac 12\left\|\tilde{U}_{-d}^T\cdot \frac{x^{\parallel}}{\|x^{\parallel}\|}\right\|^2 + \frac 12\inprod{\tilde{u}^d, \frac{x^{\parallel}}{\|x^{\parallel}\|} + \frac{x^{\perp}}{\|x^{\perp}\|}}^2\\
\nonumber {}& = \frac 12\left\|\tilde{U}_{-d}^T\cdot \frac{x^{\parallel}}{\|x^{\parallel}\|}\right\|^2 + \frac 12\left(\|x^{\parallel}\| + \|x^{\perp}\|\right)^2\\
\nonumber {}& \ge \frac 12(1-C_2\eps)^2\frac mn + \frac 12\left(\sqrt{C_4\frac dn} + \left((1-C_6\eps)^2\frac mn - C_4\frac dn\right)^{1/2}\right)^2 \\
{}& \ge \frac 12(1-C_2\eps)^2\frac mn + \frac 12\left(\sqrt{C_4\frac dn} + \left((1-C_7\eps)^2\frac mn\right)^{1/2}\right)^2 \EquationName{small-m}\\
{}&\ge (1 - C_8\eps)\frac mn + C_9\frac{\sqrt{md}}n\EquationName{finishcase2}
\end{align}
where \Equation{small-m} used that $m > cd/\eps$. Now note that for $m < cd/\eps^2$, the right hand side of \Equation{finishcase2} is at least $(1+10(C_2+1)\eps)^2m/n$ and thus $\|\tilde{U}^T x_1\| \ge (1+10(C_2+1)\eps)\sqrt{m/n}$.
\end{proof}

\section{Sparsity Lower Bound}
In this section, we consider the trade-off between $m$, the number of columns of the embedding matrix $\SE$, and $s$, the number of non-zeroes per column of $\SE$. In this section, we only consider the case $n\ge 100d^2$. By Yao's minimax principle, we only need to argue about the performance of a fixed matrix $\SE$ over a distribution over $U$. Let the distribution of the columns of $U$ be $d$ i.i.d.\ random standard basis vectors in $\R^n$. With probability at least $99/100$, the columns of $U$ are distinct and form a valid orthonormal basis for a $d$ dimensional subspace of $\R^n$. If $\SE$ succeeds on this distribution of $U$ conditioned on the fact that the columns of $U$ are orthonormal with probability at least $99/100$, then it succeeds in the original distribution with probability at least $98/100$. In section~\ref{sec:lb-eps}, we show a lower bound on $s$ in terms of $\eps$, whenever the number of columns $m$ is much smaller than $\eps^2 d^2$. In section~\ref{sec:lb-m}, we show a lower bound on $s$ in terms of $m$, for a fixed $\eps=1/2$. Finally, in section~\ref{sec:lb-both}, we show a lower bound on $s$ in terms of both $\eps$ and $m$, when they are both sufficiently small.
\subsection{Lower bound in terms of $\eps$}\label{sec:lb-eps}
\begin{theorem}
If $n\ge 100d^2$ and $m \le \eps^2 d(d-1)/32$, then $s = \Omega(1/\eps)$.
\end{theorem}
\begin{proof}
We first need a few simple lemmas.
\begin{lemma}\label{lem:nonnegative-dot-product}
Let $\mathcal{P}$ be a distribution over vectors of norm at most 1 and $u$ and $v$ be independent samples from $\mathcal{P}$. Then
$\E\inprod{u,v} \ge 0$.
\end{lemma}
\begin{proof}
Let $\delta=\E\inprod{u,v}$. Assume for the sake of contradiction that $\delta < 0$. Take $t$ samples $u_1, \ldots, u_t$ from $\mathcal{P}$. By linearity of expectation, we have $0 \le \E(\sum_i u_i)^2 \le t + t(t-1)\delta$. This is a contradiction because the RHS tends to $-\infty$ as $t\rightarrow \infty$.
\end{proof}
\begin{lemma}\label{lem:not-negative-often}
Let $X$ be a random variable bounded by $1$ and $\E X \ge 0$. Then for any $0<\delta<1$, we have $\Pr(X \le -\delta) \le 1/(1+\delta)$.
\end{lemma}
\begin{proof}
We prove the contrapositive. If $\Pr(X\le-\delta) > 1/(1+\delta)$, then 
$$\E X \le -\delta\Pr(X\le-\delta) + \Pr(X>-\delta) < -\delta/(1+\delta) + 1 - 1/(1+\delta) = 0 .$$
\end{proof}

Let $u_i$ be the $i$ column of $\SE U$, $r_i$ and $z_i$ be the index and the value of the coordinate of the maximum absolute value of $u_i$, and $v_i$ be $u_i$ with the coordinate at position $r_i$ removed. Let $p_{2j-1}$(respectively, $p_{2j}$) be the fractions columns of $\SE$ whose entry of maximum absolute value is on row $j$ and is positive (respectively, negative). Let $C_{i,j}$ be the indicator variable indicating whether $r_i=r_j$ and $z_i$ and $z_j$ are of the same sign. Let $E=\E C_{1,2} = \sum_{i=1}^{2m} p_i^2$. Let $C=\sum_{i<j\le d}C_{i,j}$. We have 
$$\E C = \frac{d(d-1)}{2}\sum_{i=1}^{2m}p_i^2 \ge \frac{d(d-1)}{4m} \ge 8\eps^{-2}$$
 If $i_1,i_2, i_3,i_4$ are distinct then $C_{i_1, i_2}, C_{i_3, i_4}$ are independent. If the pairs $(i_1, i_2)$ and $(i_3, i_4)$ share one index then $\Pr(C_{i_1,i_2}=1 \wedge C_{i_3,i_4}=1) = \sum_{i}p_i^3$ and $\Pr(C_{i_1,i_2}=1 \wedge C_{i_3,i_4}=0) = \sum_{i}p_i^2(1-p_i)$. Thus for this case,
\begin{align*}
\E(C_{i_1,i_2}-E])(C_{i_3,i_4}-E]) &=(1-2\sum_i p_i^2+\sum_i p_i^3)E^2 - 2(1-E)E\sum_i p_i^2 (1-p_i) + (1-E)^2 \sum_i p_i^3\\
&=E^2 -2E^3+E^2\sum_i p_i^3-(2E-2E^2)(E-\sum_i p_i^3) + (1-2E+E^2)\sum_i p_i^3\\
&=\sum_i p_i^3 - E^2\le \left(\sum_i p_i^2\right)^{3/2}
\end{align*}
The last inequality follows from the fact that the $\ell_3$ norm of a vector is smaller than its $\ell_2$ norm.
We have
$$\Var[C] = \frac{d(d-1)}{2}\Var[C_{1,2}] + d(d-1)(d-2)\E(C_{i_1,i_2}-\E C_{i_1,i_2})(C_{i_1,i_3}-\E C_{i_1,i_3}) \le 4(\E C)^{3/2} .$$
Therefore,
$$\Pr(C\le (\E C)/2) \le \frac{4\Var[C]}{(\E C)^2} \le O\left(\sqrt{\frac{m}{d(d-1)}}\right) .$$
Thus, with probability at least $1-O(\eps)$, we have $C\ge 4\eps^{-2}$. We now argue that there exist $1/\eps$ pairwise-disjoint pairs $(a_i, b_i)$ such that $r_{a_i} = r_{b_i}$ and $z_{a_i}$ and $z_{b_i}$ are of the same sign. Indeed, let $d_{2j-1}$ (respectively, $d_{2j}$) be the number of $u_i$'s with $r_i=j$ and $z_i$ being positive (respectively, negative). Wlog, assume that $d_1, \ldots, d_t$ are all the $d_i$'s that are at least 2. We can always get at least $\sum_{i=1}^t (d_i-1)/2$ disjoint pairs. We have
$$\sum_{i=1}^t (d_i-1)/2 \ge \frac{1}{2}\left(\sum_{i=1}^t d_i (d_i-1)/2\right)^{1/2} =\frac{C^{1/2}}{2} \ge \eps^{-1}$$

For each pair $(a_i, b_i)$, by Lemmas~\ref{lem:nonnegative-dot-product} and~\ref{lem:not-negative-often}, $\Pr[\langle v_{a_i},v_{b_i}\rangle \le -\eps] \le \frac{1}{1+\eps}$ and these events for different $i$'s are independent so with probability at least $1-(1+\eps)^{-1/\eps}\ge 1 - e^{\eps/2-1}$, there exists some $i$ such that $\langle v_{a_i},v_{b_i}\rangle > -\eps$. For $\SE$ to be a subspace embedding for the column span of $U$, it must be the case, for all $i$, that $\|u_i\| = \|\SE Ue_i\| \ge 1-\eps$. We have $|z_i|\ge s^{-1/2}\|u_i\|\ge s^{-1/2}(1-\eps)~\forall i$. Therefore, $\langle u_{a_i}, u_{b_i}\rangle \ge s^{-1}(1-\eps)^2-\eps$. We have
\begin{align*}
\left\|\SE U\left(\frac{1}{\sqrt{2}}(e_{a_i}+e_{b_i})\right)\right\|^2 &= \frac{1}{2}\|u_{a_i}\|^2 + \frac{1}{2}\|u_{b_i}\|^2 + \langle u_{a_i}, u_{b_i}\rangle\\
&\ge (1-\eps)^2(1 + s^{-1})-\eps
\end{align*}
However, $\|\SE U\|\le 1+\eps$ so $s \ge (1-\eps)^2/(5\eps)$.
\end{proof}
\subsection{Lower bound in terms of $m$}\label{sec:lb-m}
\begin{theorem}
For $n\ge 100d^2$, $\frac{20\log\log d}{\log d}<\gamma<1/12$ and $\eps=1/2$, if $m \le d^{1+\gamma}$, then $s=\Omega(1/\gamma)$.
\end{theorem}
\begin{proof}
We first prove a standard bound for a certain balls and bins problem. The proof is included for completeness.
\begin{lemma}\label{lem:uniform-balls-bins}
Let $\alpha$ be a constant in $(0,1)$.
Consider the problem of throwing $d$ balls independently and uniformly at random at $m\le d^{1+\gamma}$ bins with $\frac{10\log\log d}{\alpha\log d}<\gamma<1/12$. With probability at least $99/100$, at least $d^{1-\alpha}/2$ bins have load at least $\alpha/(2\gamma)$.
\end{lemma}
\begin{proof}
Let $X_i$ be the indicator r.v.\ for bin $i$ having $t=\alpha/(2\gamma)$ balls, and $X\eqdef \sum_i X_i$. Then
$$\E X_1 = {d\choose t}m^{-t}(1-1/m)^{d-t} \ge \left(\frac{d}{tm}\right)^t e^{-1} \ge d^{-\alpha}$$
Thus, $\E X \ge d^{1-\alpha}$.
Because $X_i$'s are negatively correlated,
$$\Var[X] \le \sum_i\Var[X_i] = n(\E X_1 - (\E X_1)^2) \le \E X .$$
By Chebyshev's inequality,
$$\Pr[X \le d^{1-\alpha}/2] \le \frac{4\Var[X]}{(\E X)^2} \le 4d^{\alpha-1}$$
Thus, with probability $1-4d^{\alpha-1}$, there exist $d^{1-\alpha}/2$ bins with at least $\alpha/(2\gamma)$ balls.
\end{proof}
Next we prove a slightly weaker bound for the non-uniform version of the problem.
\begin{lemma}\label{lem:nonuniform-balls-bins}
Consider the problem of throwing $d$ balls independently at $m\le d^{1+\gamma}$ bins. In each throw, bin $i$ receives the ball with probability $p_i$. With probability at least $99/100$, there exist $d^{1-\alpha}/2$ disjoint groups of balls of size $\alpha/(4\gamma)$ each such that all balls in the same group land in the same bin.
\end{lemma}
\begin{proof}
The following procedure is inspired by the alias method, a constant time algorithm for sampling from a given discrete distribution (see e.g.~\cite{Vose91}). We define a set of $m$ virtual bins with equal probabilities of receiving a ball as follows. The following invariant is maintained: in the $i$th step, there are $m-i+1$ values $p_1,\ldots, p_{m-i+1}$ satisfying $\sum_j p_j = (m-i+1)/m$. In the $i$th step, we create the $i$th virtual bin as follows. Pick the smallest $p_j$ and the largest $p_k$. Notice that $p_j \le 1/m \le p_k$. Form a new virtual bin from $p_j$ and $1/m-p_j$ probability mass from $p_k$. Remove $p_j$ from the collection and replace $p_k$ with $p_k+p_j-1/m$.

By Lemma~\ref{lem:uniform-balls-bins}, there exist $d^{1-\alpha}/2$ virtual bins receiving at least $\alpha/(2\gamma)$ balls. Since each virtual bin receives probability mass from at most 2 bins, there exist $d^{1-\alpha}/2$ groups of balls of size at least $\alpha/(4\gamma)$ such that all balls in the same group land in the same bin.
\end{proof}

Finally we use the above bound for balls and bins to prove the lower bound. Let $p_i$ be the fraction of columns of $\SE$ whose coordinate of largest absolute value is on row $i$. By Lemma~\ref{lem:nonuniform-balls-bins}, there exist a row $i$ and $\alpha/(4\gamma)$ columns of $\SE U$ such that the coordinates of maximum absolute value of those columns all lie on row $i$. $\SE$ is a subspace embedding for the column span of $U$ only if $\|\SE Ue_j\|\in [1/2, 3/2]~\forall j$. The columns of $\SE U$ are $s$ sparse so for any column of $\SE U$, the largest absolute value of its coordinates is at least $s^{-1/2}/2$. Therefore, $\|e_i^T \SE U\|^2 \ge \alpha/(16\gamma s)$. Because $\|\SE U\|\le 3/2$, it must be the case that $s=\Omega(\alpha/\gamma)$.

\end{proof}

\vspace{-.4in}\subsection{Combining both types of lower bounds}\label{sec:lb-both}

\begin{theorem}
For $n\ge 100d^2$, $m<d^{1+\gamma}$, $\alpha\in(0,1)$, $\frac{10\log\log d}{\alpha\log d}<\gamma<\alpha/4$, $0<\eps<1/2$, and $2/(\eps\gamma) < d^{1-\alpha}$, we must have $s=\Omega(\alpha/(\eps\gamma))$.
\end{theorem}

\begin{proof}
Let $u_i$ be the $i$ column of $\SE U$, $r_i$ and $z_i$ be the index and the value of the coordinate of the maximum absolute value of $u_i$, and $v_i$ be $u_i$ with the coordinate at position $r_i$ removed. Fix $t=\alpha/(4\gamma)$. Let $p_{2i-1}$ (respectively, $p_{2i}$) be the fractions of columns of $\SE$ whose largest entry is on row $i$ and positive (respectively, negative). By Lemma~\ref{lem:nonuniform-balls-bins}, there exist $d^{1-\alpha}/2$ disjoint groups of $t$ columns of $\SE U$ such that the columns in the same group have the entries with maximum absolute values on the same row.  Consider one such group $G=\{u_{i_1}, \ldots, u_{i_t}\}$. By Lemma~\ref{lem:nonnegative-dot-product} and linearity of expectation, $\E \sum_{u_i,u_j\in G, i\ne j}\inprod{v_i, v_j} \ge 0$. Furthermore, $\sum_{u_i,u_j\in G, i\ne j} \langle v_i, v_j \rangle\le t(t-1)$. Thus, by Lemma~\ref{lem:not-negative-often}, $\Pr(\sum_{u_i, u_j\in G, i\ne j} \inprod{v_i, v_j} \le -t(t-1)(\eps\gamma)) \le \frac{1}{1+\eps\gamma}$. This event happens independently for different groups, so with probability at least $1-(1+\eps\gamma)^{-1/(\eps\gamma)} \ge 1-e^{\eps\gamma/2-1}$, there exists a group $G$ such that
$$\sum_{u_i, u_j\in G, i\ne j} \langle v_i, v_j \rangle > -t(t-1)(\eps\gamma)$$
The matrix $\SE$ is a subspace embedding for the column span of $U$ only if for all $i$, we have $\|u_i\|=|\SE Ue_i\|\ge (1-\eps)$. We have $|z_i|\ge s^{-1/2}\|u_i\| \ge s^{-1/2}(1-\eps)$. Thus, $\sum_{u_i, u_j\in G, i\ne j}\langle u_i, u_j \rangle \ge t(t-1)((1-\eps)^2 s^{-1} - \eps\gamma)$. We have 
\begin{align*}
\left\|\SE U\left(\frac{1}{\sqrt{t}}\left(\sum_{i:u_i\in G} e_i\right)\right)\right\|^2 \ge (1-\eps)^2 + \frac{2}{t}{t\choose 2} ((1-\eps)^2 s^{-1}-\eps\gamma)\ge (1-\eps)^2(1+(t-1)s^{-1})-\alpha\eps/4
\end{align*}
Because $\|\SE U\|\le 1+\eps$, we must have $s \ge \frac{(\alpha/\gamma-4)(1-\eps)^2}{(16+\alpha)\eps}$.
\end{proof}

\bibliographystyle{plain}
\bibliography{main}

\end{document}